\documentclass[runningheads]{llncs} %

\usepackage[utf8]{inputenc}
\usepackage[T1]{fontenc}
\usepackage[english]{babel}
\usepackage[final,babel=true]{microtype}      %
\usepackage{csquotes}
\usepackage{xparse}
\usepackage{nicefrac}

\usepackage{amsmath}
\usepackage{amssymb}
\usepackage{bm}
\usepackage{textcomp}
\usepackage{booktabs}
\usepackage{multirow}
\usepackage[ruled,linesnumbered,nofillcomment]{algorithm2e}

\usepackage[pdftex]{graphicx}
\usepackage[table]{xcolor}
\usepackage{tikz}
\usepackage{floatrow}
\usepackage[unicode]{hyperref}

\newfloatcommand{capbtabbox}{table}[][\FBwidth]

\SetAlFnt{\scriptsize}
\SetAlCapFnt{\footnotesize}
\SetAlTitleFnt{\footnotesize}

\usetikzlibrary{calc}

\definecolor{C0}{HTML}{e41a1c}
\definecolor{C1}{HTML}{377eb8}

\spnewtheorem*{causalq}{Causal Question}{\itshape}{\rmfamily}

\newcommand{\xisim}{\ensuremath \mathbf{x}_i^\text{sim}}
\newcommand{\xiobs}{\ensuremath \mathbf{x}_i^\text{obs}}
\newcommand{\xihold}{\ensuremath \mathbf{x}_i^\text{holdout}}

\newcommand{\xis}{\ensuremath x_{\mathcal{S}}^{\prime}}
\newcommand{\xes}{\ensuremath x_{\bar{\mathcal{S}}}}
\newcommand{\indi}{\ensuremath \perp \!\!\!\perp}
\newcommand{\vadas}{\ensuremath \mathrm{ADAS}}
\newcommand{\vz}{\ensuremath \mathbf{z}}

\DeclareDocumentCommand{\expect}{om}{%
\ensuremath\mathbb{E}%
\IfNoValueTF {#1}{}{%
_{#1}}%
\left[ {#2} \right]
}

\begin{document} %
\title{Estimation of Causal Effects in the Presence of Unobserved Confounding
in the Alzheimer's Continuum}
\titlerunning{Estimation of Causal Effects in the Alzheimer's Continuum}
\author{Sebastian P{\"{o}}lsterl \and
    Christian Wachinger} %
\authorrunning{S.~P{\"{o}}lsterl and C.~Wachinger} %
\institute{Artificial Intelligence in Medical Imaging (AI-Med),\\
    Department of Child and Adolescent Psychiatry,\\
    Ludwig-Maximilians-Universit{\"{a}}t, Munich, Germany\\
    \email{\{sebastian.poelsterl,christian.wachinger\}@med.uni-muenchen.de}} %
\maketitle              %
\begin{abstract}
Studying the relationship between neuroanatomy and cognitive decline
due to Alzheimer's
has been a major research focus in the last decade.
However, to infer cause-effect relationships rather than simple associations
from observational data,
we need to (i) express the causal relationships leading
to cognitive decline in a graphical model,
and (ii) ensure the causal effect of interest is identifiable from
the collected data.
We derive a causal graph from the current clinical knowledge
on cause and effect in the Alzheimer's disease continuum,
and show that identifiability of the causal effect %
requires all confounders to be known and measured.
However, in complex neuroimaging studies, we neither know all potential confounders nor
do we have data on them.
To alleviate this requirement,
we leverage the dependencies among multiple causes
by deriving a substitute confounder via
a probabilistic latent factor model.
In our theoretical analysis, we prove that using the
substitute confounder enables identifiability of
the causal effect of neuroanatomy on cognition.
We quantitatively evaluate the effectiveness
of our approach on semi-synthetic data, where we
know the true causal effects, and illustrate its use
on real data on the Alzheimer's disease continuum, where it reveals
important causes that otherwise would
have been missed.
\end{abstract} %
\section{Introduction}

The last decade saw an unprecedented increase in large
multi-site neuroimaging studies, which opens the possibility
of identifying disease predictors with low-effect sizes.
However, one major obstacle to fully utilize this data is confounding.
The relationship between a measurement and an outcome is confounded if
the observed association is only due to
a third latent random variable, but there is no direct causal link
between the measurement and outcome.
If confounding is ignored, investigators will likely
make erroneous conclusions, because the observed data
distribution is compatible with many -- potentially contradictory
-- causal explanations, leaving us with no way to
differentiate between the true and false effect on the basis of
data. In this case, the causal effect is unidentifiable~\cite{Pearl2000}.

It is important to remember that what is, or is not, regarded as a
confounding variable is relative and depends on the
goal of the study.
For instance, age is often considered a confounder
when studying Alzheimer's disease (AD), but
if the focus is age-related cognitive decline in a healthy
population, age is not considered a confounder~\cite{Lockhart2014}.
Therefore, it is vital to state the causal question being studied.

Causal inference addresses confounding in a principal
manner and allows us to determine which cause-effect
relationships can be identified from a given dataset.
To infer causal effects from observational data,
we need to rely on expert knowledge and untestable
assumptions about the data-generating process to build
the causal graph linking causes, outcome, and other variables~\cite{Pearl2000}.
One essential assumption to estimate causal effects from
observational data is that of \emph{no unmeasured confounder}~\cite{Pearl2000}.
Usually, we can only identify causal effects if we
know and recorded all confounders.
However, analyses across 17 neuroimaging studies revealed
that considerable bias remains in volume and
thickness measurement after adjusting for age, gender, and
the type of MRI scanner~\cite{Wachinger2020}.
Another study on confounders in UK Biobank brain imaging
data identified hundreds of potential confounders just related to the
acquisition process researchers
would need to account for~\cite{AlfaroAlmagro2021}.
These results suggest that all factors
contributing to confounding in neuroimaging are not yet fully understood,
and hence the premise of \emph{no unmeasured confounder} is most likely
going to be violated.%

In this paper, we focus on the problem of estimating causal
effects of neuroanatomical measures on cognitive decline due to
Alzheimer's in the presence
of \emph{unobserved confounders}.
To make this feasible, we derive a causal graph
from domain knowledge on the Alzheimer's disease continuum
to capture known disease-specific relationships.
While causal affects are generally unidentifiable in the presence
of unobserved confounding, we will illustrate that we can
leverage the dependencies among multiple causes
to estimate a latent substitute confounder via a Bayesian
probabilistic factor model.
In our experiments, we quantitatively demonstrate the effectiveness
of our approach on semi-synthetic data, where we
know the true causal effects, and illustrate its use
on real data on the Alzheimer's disease continuum,
where our analyses reveal important causes of cognitive function
that would otherwise have been missed.

\paragraph{Related Work.}
Despite the importance of this topic, there has been little
prior work on causal inference for estimating causal effects in neuroimaging.
In contrast to our approach, most of the previous works assume that all confounding variables are known and have been measured.
The most common approach for confounding adjustment
is regress-out. In regress-out, the original measurement (e.g. volume or thickness)
is replaced by the residual of a regression model
fitted to estimate the original value from the confounding
variables.
In~\cite{dukart2011age}, the authors use linear regression to account
for age, which has been extended to additionally account for gender in~\cite{Koikkalainen2012}.
A Gaussian Process (GP) model has been proposed in~\cite{Kostro2014}
to adjust for age, total intracranial volume, sex, and MRI scanner.
Fortin~et~al.~\cite{fortin2018harmonization} proposed a linear mixed effects model
to account for systematic differences across imaging sites.
This model has been extended in~\cite{Wachinger2020} to account for both observed
and unobserved confounders.
In~\cite{Snoek2019}, linear regression is used to regress-out the
effect of total brain volume.
In addition to the regress-out approach, analyses can be adjusted for confounders
by computing instance weights, %
that are used in a downstream classification or
regression model to obtain a pseudo-population that is approximately balanced
with respect to the confounders~\cite{linn2016addressing,rao2017predictive}.
A weighted support vector machine to adjust for age
is proposed in~\cite{linn2016addressing}.
In \cite{rao2017predictive}, weighted GP regression is proposed
to account for gender and imaging site effects.
We note that none of the work above studied whether causal
effects can actually be identified from observed data using
the theory of causal inference.

\section{Methods}

Causal inference from observational data comprises multiple steps,
(i) defining the causal question and its associated causal graph,
(ii) determining under which conditions the question can be answered
from real world data,
and (iii) estimating the causal effects via modelling.
We denote random variables with uppercase letters and specific values taken
by the corresponding variables with lowercase letters.
We distinguish between real-valued subcortical volume ($X_1^v,\ldots,X_{D_1}^v$)
and cortical thickness ($X_1^t,\ldots,X_{D_2}^t$) measurements. We denote by
$X_1,\ldots,X_D$ all measurements, irrespective of their type
($D = D_1 + D_2$).
Next, we will specify our causal question and determine when
the causal effect of a subset of measurements on an outcome
is identifiable using Pearl's do-calculus~\cite{Pearl2000}.

\subsection{The Causal Question and Its Associated Graph}
\begin{causalq}
    What is the average causal effect of
    increasing/decreasing the volume or thickness of a subset of
    neuroanatomical structures on the
    Alzheimer's Disease Assessment Scale
    Cognitive Subscale 13 score (ADAS; \cite{Mohs1997}) in
    patients with an Alzheimer's pathologic change~\cite{Jack2018}?
\end{causalq}

The gold standard to answer this question would be
a randomized experiment, where subjects'
volumes and thicknesses are randomly assigned.
As this is impossible, we have to resort to observational data.
To estimate causal effects from observational data,
we need to rely on expert knowledge to build
the causal graph linking causes, outcome, and other variables~\cite{Pearl2000}.
Fig.~\ref{fig:causal_graph} depicts the graph related to our causal question.
We explain our reasoning below.

Our causal question already determines that the causal graph needs to
comprise ADAS (the outcome),
measures $X_1,\ldots,X_D$ (the causes),
and the level of beta amyloid 42 peptides (A$\beta$), which determines
whether a patient has an Alzheimer's pathologic change~\cite{Jack2018}.
To link A$\beta$ with the remaining variables, we rely on expert
knowledge, namely that $A\beta$ causes levels of Tau
phosphorylated at threonine 181 (p-Tau), which in turn
causes neurodegeneration that ultimately results
in cognitive decline~\cite{Jack2018,Jack2013}.
Moreover, we consider that the patient's
levels of A$\beta$ and p-Tau are determined by the
allelic variant of apolipoprotein E (ApoE; \cite{Long2019}), among other
unobserved common causes (dashed line), and
that aging influences the neuroanatomy, amyloid and tau pathology~\cite{Barnes2010,Crary2014,Rodrigue2009}.
Beside biological relationships, we also include demographic
and socio-economic factors.
In particular, p-Tau levels, brain size, and the level
of education is known to differ in males and females~\cite{Ferretti2020}.
We model resilience to neurodegeneration by including
years of education (Edu) as a proxy for cognitive reserve,
and total intracranial volume (TIV) as a proxy for brain reserve,
where the latter is only causal for volume measurements
$X_1^v,\ldots,X_{D_1}^v$~\cite{Stern2020}.
Finally, fig.~\ref{fig:causal_graph} reflects that
age is a known confounder of the relationship between neuroanatomy and cognition
~\cite{Hedden2004}.
However, as outlined in the introduction, the full set of confounders
is extensive and most of them are unmeasured.
Therefore, we assume an unknown and
unobserved set of additional confounders ${U}$.

\begin{figure}[tb]
    \centering
    \includegraphics[scale=0.465]{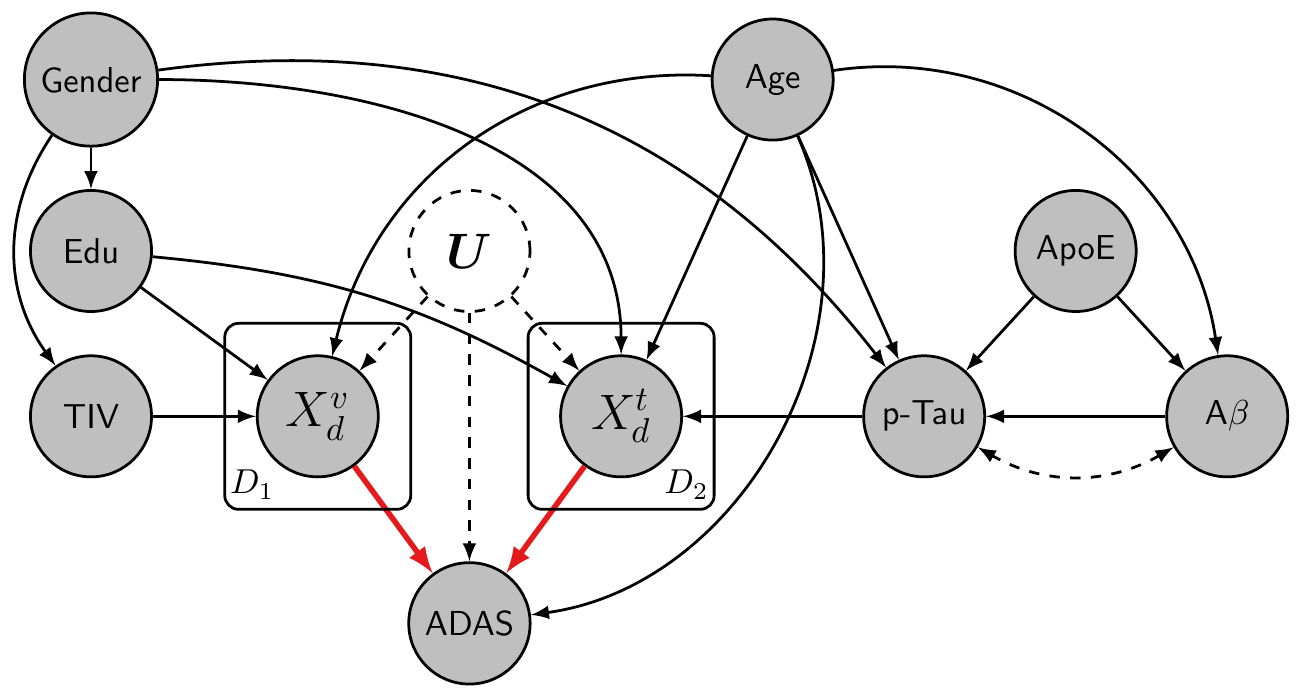}%
    \caption{\label{fig:causal_graph}%
    Causal graph used to estimate the causal effect (red arrow) of
    subcortical volume ($X_d^v$) and cortical thickness ($X_d^t$)
    on cognitive function (ADAS)
    in the presence of an unknown and unobserved confounder
    $U$.
    Exogenous variables irrelevant
    for estimating the causal effect of interest are not shown.
    Circles are random variables and
    arrows causal relationships. Filled circles are observed,
    transparent circles are hidden,
    bidirectional edges denote unobserved common causes.
    }
\end{figure}

\subsection{Identifiability in the Presence of an Unobserved Confounder}
Formally, the causal question states that we want to
use the observed data to
estimate the average causal effect
that a subset
$\mathcal{S} \subset \{X_1,\ldots,X_D\}$
of neuroanatomical structures
have simultaneously on the ADAS score:
\begin{equation}\label{eq:ace}
    \expect{\vadas\,|\, do(X_{\mathcal{S}} = \xis)}
    = \int \text{adas} \cdot P(\text{adas}\,|\,do(\xis))\,d\text{adas} .
\end{equation}
The do-operator states that we are interested
in the post-intervention distribution of
ADAS, induced by the intervention that sets the neuroanatomical measures
$X_{\mathcal{S}}$ equal to $\xis$.
The central question in causal inference is that
of identification: Can the post-intervention distribution
$P(\text{adas}\,|\,do(x))$ be estimated from data governed by the
observed joint distribution over $X$ and $\vadas$?

Inspecting fig.~\ref{fig:causal_graph} reveals that the
relationship between any $X_d$ and $\vadas$ is confounded by the known
confounder age, but also the unknown confounder $U$.
Knowing that the causal effect of $X_d$ on $\vadas$ is unidentifiable
in the presence of unobserved confounding, it initially appears
that our causal question cannot be answered~\cite[Theorem 3.2.5]{Pearl2000}.
Next, we will detail how we can resolve this issue by
exploiting the influence that the unknown confounder has on multiple
causes simultaneously.
The full procedure is outlined in algorithm~\ref{algo:deconfounder}.

\subsection{Estimating a Substitute Confounder}

By assuming that the confounder $U$ is unobserved,
we can only attempt to estimate causal effects
by building upon assumptions on the data-generating process.
Therefore, we assume that the data-generating process is
faithful to the graphical model in fig.~\ref{fig:causal_graph},
i.e., statistical independencies in the
observed data distribution imply missing causal relationships in
the graph. %
In particular, this implies that there is a common unobserved
confounder $U$ that is shared among all causes $X$ and
that there is no unobserved confounder that affects a single cause.
Here,
causes $X$ are image-derived volume and thickness measurements,
and we require that the unknown confounder affects
multiple brain regions and not just a single region.
This assumption is plausible, because common sources of confounding
such as scanner, protocol, and aging affect the brain
as a whole and not just individual regions~\cite{Barnes2010,Stonnington2008}.
Based on this assumption, we can exploit the fact that
the confounder induces dependence among multiple causes.

\begin{algorithm}[tb]
    \DontPrintSemicolon
    \caption{\label{algo:deconfounder}
    Causal inference with unobserved confounding.}
    \KwIn{%
        Neuroanatomical measures $\mathbf{X} \in \mathbb{R}^{N \times D}$,
        direct influences on neuroanatomy
        $\mathbf{F} \in \mathbb{R}^{N \times P}$,
        ADAS scores
        $\mathbf{y} \in \mathbb{R}^N$;
        values for intervention $\xis$;
        $\tau$ minimum Bayesian p-value for model checking.
    }
    \KwOut{Estimate of $\expect{\vadas\,|\, do(X_{\mathcal{S}} = \xis)}$.}
    \BlankLine
    Sample a random binary matrix $\mathbf{H} \in \{0; 1\}^{N \times D}$ and
    split data into $\mathbf{X}^\text{obs} = (1 - \mathbf{H}) \odot \mathbf{X}$
    and $\mathbf{X}^\text{holdout} = \mathbf{H} \odot \mathbf{X}$. \;
    Fit PLFM with parameters $\theta$ to $\mathbf{X}^\text{obs}$ and $\mathbf{F}$. \;
    Simulate $M$ replicates of $\mathbf{x}_i$ by drawing from posterior predictive distribution $P(\xisim\,|\,\xiobs) =
    \int P(\xisim\,|\,\theta) P(\theta\,|\,\xiobs)\,d\theta$. \;
    For each observation, estimate Bayesian p-value using test statistic in \eqref{eq:bayesian-p-value}:
    $p_{B_i} \approx \frac{1}{M} \sum_{m=1}^M I(
        T(\mathbf{x}_{i,m}^\text{sim}) \geq T(\xihold))$. \;
    \If{$\frac{1}{N} \sum_{i=1}^N p_{B_i} > \tau$}
    {
        Estimate substitute confounders $\hat{\mathbf{Z}} = \expect{\mathbf{Z}\,|\,\mathbf{X}^\text{obs}, \mathbf{F}}$ by PLFM. \;
        Fit a regression model $f\colon \mathbf{x} \mapsto \vadas$,
        using the residuals defined in \eqref{eq:residuals}.
         \;
        $\expect{\vadas\,|\, do(X_{\mathcal{S}} = \xis)} \approx \frac{1}{N} \sum_{i=1}^N f(r(\tilde{\mathbf{x}}_i, \mathbf{f}_i))$, where $\tilde{\mathbf{x}}_i$ equals
        $\mathbf{x}_i$, except for features in $\mathcal{S}$, which are set to $\xis$. \;
    }
\end{algorithm}

From fig.~\ref{fig:causal_graph} we can
observe that $U$, age, education, gender, and p-Tau are
shared among all causes $X_1,\ldots,X_D$.
Given these parents, denoted as $PA_{X_1,\ldots,X_D}$, the causes become conditionally independent:
\begin{equation}\label{eq:condprob}%
    P(x_1, \ldots x_D\,|\,PA_{X_1,\ldots,X_D})
    = {\textstyle
    \prod_{d=1}^D P(x_d\,|\,PA_{X_1,\ldots,X_D})
    }.
\end{equation}
The key realization of our proposed method is that the conditional
probability \eqref{eq:condprob}, which is derived solely from the causal graph
in fig.~\ref{fig:causal_graph}, has the same form
as the conditional distribution of a probabilistic latent factor model (PLFM).
Therefore, we can utilize this connection to estimate a substitute
confounder $\vz$ for the unobserved confounder $U$ via a latent factor model.

The theoretical proof for this approach is due to
Wang and Blei~\cite{Wang2019} who showed that
the latent representation of any PLFM does indeed
render the relationship between neuroanatomical measures
and ADAS unconfounded
if (i) the PLFM captures all multi-cause confounders,
and (ii) the PLFM estimates
the substitute confounder with consistency,
i.e., deterministicly, as the number of causes $D$ grows large.
To verify whether (i) holds, we rely on posterior predictive checking,
as described below, and only proceed with estimation of the causal effect if the check
passes.\footnote{%
We follow~\cite{Wang2019} and use $\bar{p}_B = \expect{ p_{B_i} } > 0.1$ as criterion.}
Regarding (ii), \cite{Chen2019} showed that estimates of many PLFM are
consistent if the number of causes and samples is large.
Note, that it does not imply that
we need to find the true confounder, just a deterministic
bijective transformation of it~\cite{Wang2019}.

\paragraph{Latent Factor Model.}
Let $\mathbf{f}_i \in \mathbb{R}^P$
be the feature vector describing the observed direct influences on $X_d$
for the $i$-th patient, except TIV, which we
account for by dividing all volume measures $X_1^v,\ldots,X_{D_1}^v$ by it.
Then, we can use an extended version of
probabilistic principal component analysis~(PPCA, \cite{Tipping1999})
to represent the $D$ causes in terms of the known causes $\mathbf{f}_i$
and the latent substitute confounder $\mathbf{z}_i \in \mathbb{R}^K$:
\begin{equation}\label{eq:ppca}
    \mathbf{x}_{i} = \mathbf{W}\mathbf{z}_i + \mathbf{A} \mathbf{f}_i + \bm{\varepsilon}_i,
    \quad \bm{\varepsilon}_i \sim \mathcal{N}(\mathbf{0}, \sigma_x^2 \mathbf{I}_D),
    \qquad \forall i = 1,\ldots,N,
\end{equation}
where $\mathbf{I}_D$ is a $D \times D$ identity matrix,
$\mathbf{W}$ a $D \times K$ loading matrix,
and $\mathbf{A}$ a $D \times P$ matrix of regression coefficients
for the known causes of $\mathbf{x}_i$, excluding TIV.

Our approach is not restricted to PPCA, in fact, any PLFM can be used
to infer the substitute confounder.
Here, we consider an extended version of
probabilistic matrix factorization~(BPMF, \cite{Salakhutdinov2008}) as an
alternative:
\begin{equation}\label{eq:bpmf}
    x_{ij} = \mathbf{z}_i^\top \mathbf{v}_j + \mathbf{A}_j^\top \mathbf{f}_{i} + \varepsilon_{ij},
    \quad \varepsilon_{ij} \sim \mathcal{N}(0, \sigma_x^2),
    \quad \forall i = 1,\ldots,N,
    ~\forall j = 1,\ldots,D,
\end{equation}
where $\mathbf{v}_{j}$ is a $K$-dimensional feature-specific
latent vector.
The full models with prior distributions are
depicted in fig.~\ref{fig:latent-factor-models}.

\paragraph{Posterior predictive checking.}
To ensure the PLFM can represent the joint
distribution over the observed causes well,
we employ posterior predictive checking to quantify how well
the PLFM fits the data~\cite[ch.~6]{Gelman2013}.
If the PLFM is a good fit, simulated data generated under the PLFM
should look similar to observed data.
First, we hold-out a randomly selected portion of the observed causes,
yielding $\mathbf{X}^\text{obs}$ to fit the factor model,
and $\mathbf{X}^\text{holdout}$ for model checking.
Next, we draw simulated data from the joint posterior predictive distribution.
If there is a systematic difference between the simulated and
the held-out data, we can conclude that the PLFM
does not represent the causes well.
We use the expected negative log-likelihood as test statistic
to compute the Bayesian p-value $p_B$
-- the probability
that the simulated is more extreme than the observed data~\cite{Gelman2013}:
\begin{equation}\label{eq:bayesian-p-value}
    T(\mathbf{x}_i) = \expect[\theta]{-\log p(\mathbf{x}_i\,|\,\theta\,)|\, \xiobs},
    \quad
    p_{B_i} = P(T(\xisim) \geq T(\xihold)\,|\,\xiobs) ,
\end{equation}
where $\theta$ is the set of all parameters of the PLFM.
We estimate $p_{B_i}$ by drawing $\xisim$ repeatedly
from the posterior predictive distribution and computing the proportion
for which $T(\xisim) \geq T(\xihold)$ (see algorithm~\ref{algo:deconfounder}).
Next, we will prove that the causal effect of neuroanatomical measures
on ADAS is identifiable by accounting for the substitute confounder.

\begin{figure}[t]
    \centering%
    \includegraphics[width=.99\textwidth]{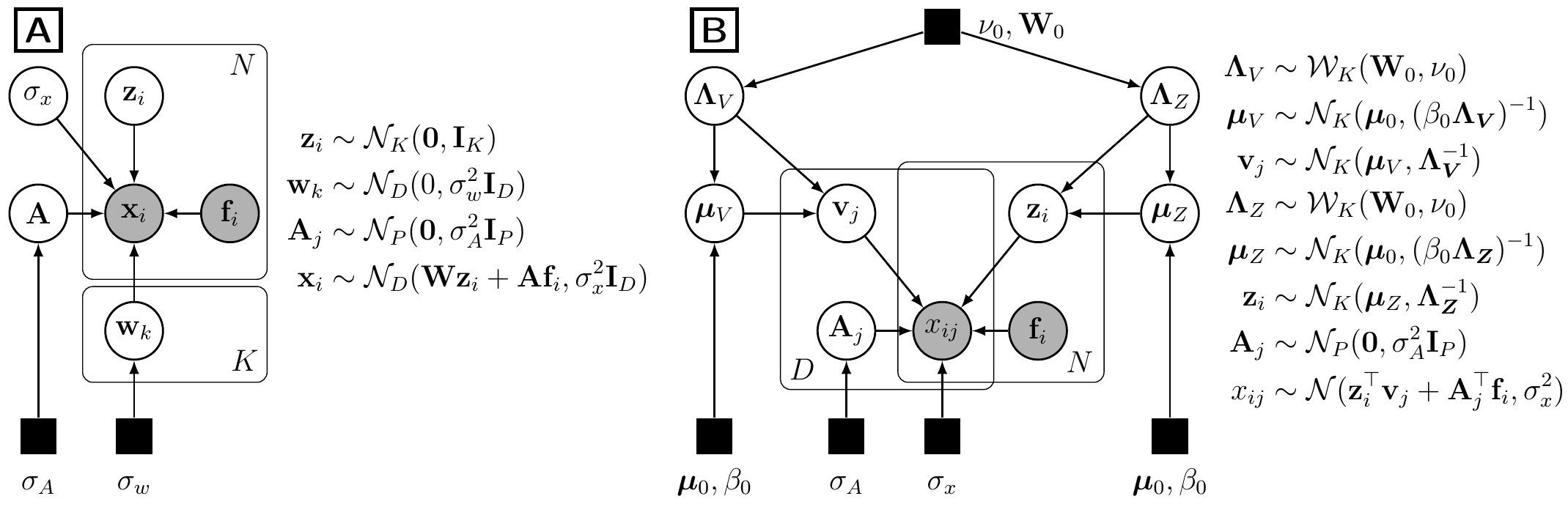}%
    \caption{Probabilistic latent factor models
    to estimate a $K$-dimensional
    substitute confounder $\mathbf{z}_i$ from $D$-dimensional causes $\mathbf{x}_i$.
    Circles are random variables,
    filled circles are observed,
    transparent circles are to be estimated.
    A: Probabilistic principal component analysis model.
    B: Bayesian probabilistic matrix factorization model.
    \label{fig:latent-factor-models}}
\end{figure}

\subsection{Identifiability in the Presence of a Substitute Confounder}

The theoretical results outlined above and described in detail in~\cite{Wang2019}
enable us to treat the substitute confounder
$\vz$ as if it were observed.
We will now prove that the average causal effect \eqref{eq:ace} that
a subset $\mathcal{S}$ of neuroanatomical structures have simultaneously
on the ADAS score is identifiable in this modified setting.
Therefore, we will again refer to fig.~\ref{fig:causal_graph}, but
replace $U$ with its observed substitute.

\begin{theorem}
    Assuming the data-generating process is
    faithful to the graphical model in fig.~\ref{fig:causal_graph},
    the causal effect
    $\expect{\vadas\,|\,do(\xis)}$ of a subset $\mathcal{S}$
    of neuroanatomical measures
    on ADAS is
    identifiable from a distribution over the observed
    neuroanatomical measures not in $\mathcal{S}$,
    age, and the substitute confounder $\vz$.
\end{theorem}
\begin{proof}
    By assuming that the graph in fig.~\ref{fig:causal_graph}
    is faithful, we can apply the rules of do calculus~\cite[Theorem 3.4.1]{Pearl2000}
    to show that the post-intervention distribution can be
    identified from observed data and the substitute confounder.
    We denote by $\bar{\mathcal{S}}$ the complement of the set $\mathcal{S}$:
    \begin{align}
    \label{eq:ace-proof-1}
    \expect{\vadas\,|\,do(\xis)}
    =& \expect[age,\xes,z]{ \expect{\vadas\,|\,do(\xis),\xes,age,\vz} } \\
    \label{eq:ace-proof-2}
    =& \expect[age,\xes,z]{ \expect{\vadas\,|\,do(\xis),\xes,do(ptau),age,\vz} } \\
    \label{eq:ace-proof-3}
    =& \expect[age,\xes,z]{ \expect{\vadas\,|\,\xis,\xes,do(ptau),age,\vz} } \\
    \label{eq:ace-proof-4}
    =& \expect[age,\xes,z]{ \expect{\vadas\,|\,\xis,\xes,ptau,age,\vz} } \\
    \label{eq:ace-proof-5}
    =& \expect[age,\xes,z]{ \expect{\vadas\,|\,\xis,\xes,age,\vz} } \\
    \label{eq:ace-proof-6}
    \approx& {\textstyle
    \frac{1}{N} \sum_{i=1}^N
    \hat{\mathbb{E}}\left[
        \vadas\,|\,\xis,\mathbf{x}_{i,\bar{\mathcal{S}}},age_{i},\vz_{i}
    \right]} .
    \end{align}
    The equality in \eqref{eq:ace-proof-1} is due to the factorization given
    by the graph in fig.~\ref{fig:causal_graph},
    the one in \eqref{eq:ace-proof-2} due to rule 3 of do calculus,
    \eqref{eq:ace-proof-3} and \eqref{eq:ace-proof-4} are due
    to rule 2 of do calculus,
    and \eqref{eq:ace-proof-5} is due to
    $\vadas \indi \text{p-Tau} \mid \xis,\xes,age,\vz$.
    Finally, we can estimate the outer expectation by Monte Carlo
    and the inner expectation with a regression model
    from the observed data alone, if
    $P(x_\mathcal{S}\,|\,PA_{X_1,\ldots,X_D}) > 0$ for any subset $S$.
    This assumption holds for the proposed PLFM
    in \eqref{eq:ppca} and \eqref{eq:bpmf}, because
    their conditional distribution is a normal distribution, which
    is non-zero everywhere.
\end{proof}

\subsection{The Outcome Model}

Theorem 1 tells us that the
average causal effect can be estimated from the observed data
and the substitute confounder.
The final step in causal inference is to actually
estimate the expectation in \eqref{eq:ace-proof-6}, which
we do by fitting a model to predict
the ADAS score from neuroanatomical measures, age and the substitute confounder
(see algorithm~\ref{algo:deconfounder}).
We will use a linear model, but Theorem 1 holds when using a non-linear model too.

ADAS ranges between 0 and 85 with higher values indicating a higher
cognitive decline, hence
we convert ADAS to proportions in the interval $(0, 1)$
and use a Bayesian Beta regression model for prediction~\cite{Ferrari2004}.
Let $y_i$ denote the ADAS score of subject $i$, then the likelihood function is
\begin{equation*}
    L(\beta_0, \bm{\beta}, \phi)
    = \prod_{i=1}^N \frac{
        y_i^{\mu_i\phi-1}(1-y_i)^{(1-\mu_i)\phi-1}
    }{
        \mathrm{B}(\mu_i\phi,(1-\mu_i)\phi)
    } ,
    \quad
    \mu_i = \mathrm{logit}^{-1}(
        \beta_0
        + r(\mathbf{x}_i, \mathbf{f}_i)^\top \bm{\beta}
    ) ,
\end{equation*}
where B is the beta function,
and $\phi$ is a scalar scale parameter.
To account for unobserved confounding, we replace the original
neuroanatomical measures by the residuals with respect to
their reconstruction by the PLFM:
\begin{equation}\label{eq:residuals}
    r(\mathbf{x}_i, \mathbf{f}_i) = \mathbf{x}_i - \expect{X_1,\ldots,X_D \mid \hat{\mathbf{z}}_i, \mathbf{f}_i}, \qquad
    \hat{\mathbf{z}}_i = \expect{Z\,|\,\mathbf{x}_i^\text{obs}, \mathbf{f}_i} ,
\end{equation}
where the expectations are with respect to the selected PLFM.\footnote{Code available at \url{https://github.com/ai-med/causal-effects-in-alzheimers-continuum}}

\section{Experiments}

\begin{table}[tb]
    \centering%
    \caption{\label{tab:synthetic-ukb}%
    $\text{RMSE} \times 100$ of effects estimated by logistic regression
    compared to the true causal effects on semi-synthetic data.
    ROA is the error when only regressing out the observed
    confounder age,
    Oracle is the error when including all confounders.
    Columns with $\Delta$ denote
    the improvement over the `Non-causal' model.}%
    \begin{scriptsize}
    \begin{tabular}{r}
        \\
        \multirow{5}{*}{%
        \begin{tikzpicture}[baseline=0,every node/.style={inner sep=0pt}]
            \node[font=\scriptsize,rotate=90,text width=1.5cm,align=right]
            (txt) {least\\confounded};
            \draw[-latex] (txt.south east)++(1mm,-1mm) -- ++(3mm,0);
        \end{tikzpicture}}
        \\
        \\
        \\
        \\
        \\
        \\
        \\
        \\
        \multirow[b]{5}{*}{%
        \begin{tikzpicture}[baseline=0,every node/.style={inner sep=0pt}]
            \node[font=\scriptsize,rotate=90,text width=1.5cm,xshift=1.5ex]
                (txtb) {most\\confounded};
            \draw[-latex] (txtb.south west)++(1mm,1mm) -- ++(3mm,0);
        \end{tikzpicture}}
        \\
        \\
        \\
        \\
        \\
        \\
    \end{tabular}
    \rowcolors{2}{}{gray!20}%
    \setlength{\tabcolsep}{3pt}
    \begin{tabular}{lrrrrrrrr}
        \toprule
        \nicefrac{$\nu_x$}{$\nu_z$}&  Non-causal & ROA & PPCA&  BPMF &  Oracle & $\Delta$ROA & $\Delta$PPCA & $\Delta$BPMF\\
        \midrule
        \nicefrac{10}{1} &     20.121 &       19.416 &  17.449 &  18.033 &  17.917 &                0.705 &         2.672 &         2.087 \\
        \nicefrac{5}{1}  &     21.078 &       20.436 &  18.560 &  18.895 &  18.781 &                0.643 &         2.518 &         2.183 \\
        \nicefrac{4}{1}  &     21.505 &       20.889 &  19.050 &  19.296 &  19.169 &                0.617 &         2.456 &         2.210 \\
        \nicefrac{3}{1}  &     22.169 &       21.590 &  19.820 &  19.933 &  19.782 &                0.580 &         2.349 &         2.236 \\
        \nicefrac{5}{2}  &     22.653 &       22.097 &  20.382 &  20.408 &  20.233 &                0.556 &         2.270 &         2.244 \\
        \nicefrac{5}{3}  &     23.911 &       23.417 &  21.837 &  21.699 &  21.438 &                0.494 &         2.073 &         2.212 \\
        \nicefrac{3}{2}  &     24.275 &       23.798 &  22.261 &  22.083 &  21.796 &                0.477 &         2.014 &         2.192 \\
        \nicefrac{1}{1}  &     25.802 &       25.391 &  24.022 &  23.735 &  23.319 &                0.411 &         1.780 &         2.067 \\
        \nicefrac{2}{3}  &     27.464 &       27.116 &  25.932 &  25.596 &  25.040 &                0.348 &         1.531 &         1.867 \\
        \nicefrac{3}{5}  &     27.899 &       27.567 &  26.434 &  26.088 &  25.502 &                0.333 &         1.465 &         1.811 \\
        \nicefrac{2}{5}  &     29.564 &       29.284 &  28.354 &  28.018 &  27.307 &                0.280 &         1.210 &         1.546 \\
        \nicefrac{1}{3}  &     30.297 &       30.037 &  29.199 &  28.877 &  28.124 &                0.259 &         1.097 &         1.419 \\
        \nicefrac{1}{4}  &     31.384 &       31.153 &  30.448 &  30.167 &  29.354 &                0.230 &         0.936 &         1.216 \\
        \nicefrac{1}{5}  &     32.179 &       31.967 &  31.364 &  31.118 &  30.266 &                0.212 &         0.815 &         1.060 \\
        \nicefrac{1}{10} &     34.318 &       34.154 &  33.816 &  33.699 &  32.778 &                0.164 &         0.502 &         0.620 \\
        \bottomrule
        \end{tabular}
    \end{scriptsize}
\end{table}

\subsubsection{Semi-synthetic Data.}

In our first experiment, we evaluate how well causal effects
can be recovered when ignoring all confounders,
using only the observed confounder (age), using the observed and unobserved
confounders (oracle), and using the observed and substitute confounder computed by
\eqref{eq:ppca} or \eqref{eq:bpmf}.
We use T1-weighted magnetic resonance imaging brains scans
from $N$\,=\,11,800 subjects from UK Biobank~\cite{miller2016multimodal}.
From each scan, we extract 19 volume measurements
with FreeSurfer~5.3~\cite{Fischl2012}
and create a synthetic binary outcome.
For each volume, we use age and gender as known causes $\mathbf{f}_i$
and estimate $\hat{x}_{ij} = \expect{X_{ij}\,|\,\mathbf{f}_i}$
via linear regression.
We use age as an observed confounder and
generate one unobserved confounder $u_k$ by
assigning individuals to clusters with varying percentages
of positive labels.
First, we obtain the first two principal components across
all volumes, scale individual scores to $[0; 1]$, and
cluster the projected data into 4 clusters using $k$-means
to assign $u_k \in \{1,2,3,4\}$ and scale $\sigma_k$ of the noise term.
Causal effects follow a sparse normal distribution ($\mathcal{N}_\text{sp}$),
where all values in the 20--80th percentile range are zero, hence
only a small portion of volumes have a non-zero causal effect.
Let $\nu_x$, $\nu_z$, $\nu_\varepsilon = 1- \nu_x - \nu_z$
denote how much variance can be explained by the
causal effects, confounding effects, and the noise, respectively,
then for the $i$-th instance in cluster $k$, we generate a binary outcome $y_i$ as:
\begin{equation*}
\begin{split}
    y_i &\sim \mathrm{Bernoulli}(%
    {\textstyle
        \mathrm{logit}^{-1}(
        \beta_0 +
        \hat{\mathbf{x}}_i \bm{\beta} \frac{\sqrt{\nu_x}}{\sigma_x}
        + u_k \cdot \frac{\sqrt{0.9 \nu_z}}{\sigma_z}
        + \text{age}_i \gamma \cdot \frac{\sqrt{0.1 \nu_z}}{\sigma_\text{age}}
        + \varepsilon_i \frac{\sqrt{\nu_\varepsilon}}{\sigma_\varepsilon}
    )}), \\
    \beta_j &\sim \mathcal{N}_\text{sp}(0, 0.5), \quad
    \gamma \sim \mathcal{N}(0, 0.2), \quad
    \sigma_k \sim 1 + \mathrm{InvGamma}(3, 1), \quad
    \varepsilon_i \sim \mathcal{N}(0, \sigma_k),
\end{split}
\end{equation*}
where $\sigma_x$, $\sigma_z$, $\sigma_\text{age}$, and $\sigma_\varepsilon$ are standard deviations
with respect to $\hat{\mathbf{x}}_i \bm{\beta}$, $u$, $\text{age}_i$ and $\varepsilon_i$ for
$i=1,\ldots,N$.
Finally, we choose $\beta_0$ such that the
positive and negative class are roughly balanced.

Table \ref{tab:synthetic-ukb} shows the root mean squared error (RMSE)
with respect
to the true causal effect of a logistic regression model
across 1,000 simulations
for various $\nu_x$, $\nu_z$,
and $\nu_\varepsilon=0.1$.
We used $K=5$ substitute confounders and
both PLFM
passed the posterior predictive check with
$\bar{p}_B = 0.269$ (BPMF) and $\bar{p}_B = 0.777$ (PPCA), despite that
the true data generation model differs.
As expected, by ignoring confounding completely (first column), the RMSE is the highest.
When only accounting for the known confounder age (second column),
the RMSE decreases slightly. The RMSE reduces considerably
when using a substitute confounder and achieves an improvement
2.9 -- 5.5 higher than that of the age-only model.
Finally, the results show that there is a cost to using a substitute confounder:
using all confounders (Oracle) leads to the lowest RMSE.

\subsubsection{Alzheimer's Disease Data.}
In this experiment, we study the causal effect of
neuroanatomical measures on ADAS using
data from the Alzheimer's Disease Neuroimaging
Initiative~\cite{jack2008alzheimer}.
We only focus on effects due to
Alzheimer's pathologic change and not other forms of dementia.
Therefore, we only include patients with abnormal amyloid biomarkers~\cite{Jack2018}.
We extract 14 volume and 8 thickness measures using
FreeSurfer~\cite{Fischl2012} for 711 subjects
(highly correlated measures are removed).
Since the average causal effect is fully parameterized
by the coefficients of the linear Beta regression model,
we can compare estimated coefficients
of the proposed approach
with $K=6$ substitute confounders,
with that of a model ignoring all confounders
(Non-causal), and of a model trained on measures where
age, gender, and education has been regressed-out.

\begin{figure}[tb]
    \centering
    \includegraphics[width=\textwidth]{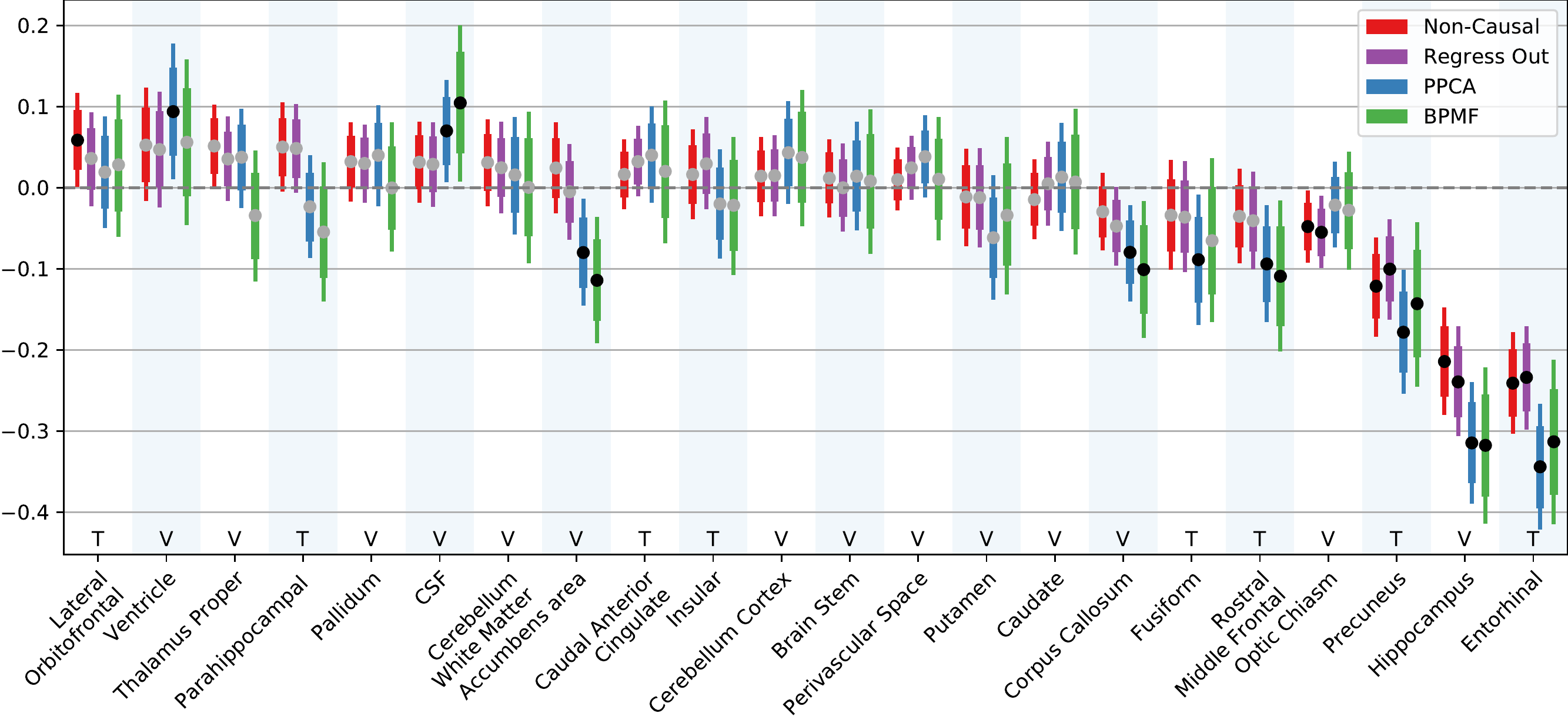}%
    \caption{\label{fig:coef-adni}%
    Mean coefficient (dot), 80\% (thick line), and 95\% (thin line)
    credible interval of \underline{v}olume and \underline{t}hickness %
    measures. Significant effects are marked with a black dot.}%
\end{figure}

The BPMF ($\bar{p}_B = 0.293$) and PPCA model ($\bar{p}_B = 0.762$)
passed the posterior predictive check;
estimated coefficients are depicted in fig.~\ref{fig:coef-adni}.
Lateral orbitofrontal thickness and optic chiasm
become non-significant after correcting for unobserved
confounding, whereas
rostral anterior cingulate thickness,
CSF volume,
accumbens volume, and
corpus callosum volume
become significant.
The biggest change concerns accumbens volume, which
is associated with cognitive \emph{improvement} in the
non-causal model, but is a cause for cognitive \emph{decline}
in the causal models.
This correction is justified, because
the accumbens is part of the limbic circuit
and thus shares its vulnerability to degenerate
during cognitive decline~\cite{Jong2012}.
An analysis based on the non-causal model would have
resulted in a wrong conclusion.
The result that atrophy of corpus callosum is causal
seems to be plausible too,
because it is a known marker of the progressive
interhemispheric disconnection in AD~\cite{Delbeuck2003}.
While finding literature on the absence of an
effect on cognition is difficult,
we believe a causal effect of atrophy of the optic nerve
to be unlikely, because
its main function is to transmit visual information.
It is more likely that the non-causal model
picked up an association due to aging-related confounding instead.
Finally, we want to highlight the change in parahippocampal thickness.
Previous research suggests that thinning is common in AD~\cite{Schwarz2016},
which would only be captured correctly after correcting for unobserved
confounding, as indicated by a negative mean coefficient.
In contrast, when only accounting for the known confounder age,
the estimated mean coefficient remains positive.

\section{Conclusion}

Inferring causal effects from observational neuroimaging data is challenging,
because it requires expressing the causal question in a graphical model,
and various unknown sources of confounding often render the causal effects
unidentifiable.
We tackled this task by deriving a causal graph from the current clinical
knowledge on the Alzheimer's disease continuum, and
proposed a latent factor model approach to
estimate a substitute for the
unobserved confounders.
Our experiments on semi-synthetic data showed that our
proposed approach can recover causal effects more accurately than a model
that only accounts for observed confounders or ignores confounding.
Analyses on the causal effects on cognitive decline
due to Alzheimer's
revealed that accounting for unobserved confounding
reveals important causes of cognitive decline that
otherwise would have been missed.

\subsubsection*{Acknowledgements.}
This research was supported by the Bavarian State Ministry of Science and the Arts and coordinated by the Bavarian Research Institute for Digital Transformation,
and the Federal Ministry of Education and Research in the call for Computational Life Sciences (DeepMentia, 031L0200A).
\bibliographystyle{splncs04}
\bibliography{references}

\end{document}